\documentclass[runningheads,a4paper]{llncs}

\usepackage{amssymb}
\usepackage{graphicx}
\usepackage{color}

\usepackage{url}

\newcommand{\keywords}[1]{\par\addvspace\baselineskip
\noindent\keywordname\enspace\ignorespaces#1}

\newtheorem{algorithm}[theorem]{Algorithm}

\begin{document}

\mainmatter  

\title{An Algorithm Computing the Core of a K\"{o}nig-Egerv\'{a}ry Graph}

\titlerunning{Core of a K\"{o}nig-Egerv\'{a}ry Graph}

%
%
\author{Vadim E. Levit\inst{1} \and Eugen Mandrescu\inst{2}}
\authorrunning{Levit and Mandrescu}

\institute{Ariel University Center of Samaria, ISRAEL \\ \email{levitv@ariel.ac.il}
\and Holon Institute of Technology, ISRAEL\\ \email{eugen\_m@hit.ac.il}
\url{ }}

%
%

\toctitle{Lecture Notes in Computer Science}
\tocauthor{Authors' Instructions}
\maketitle

\begin{abstract}
A set $S$ of vertices is \textit{independent} (or \textit{stable}) in a graph $G$ if no two
vertices from $S$ are adjacent, and $\alpha(G)$ is the cardinality of a largest
(i.e., \textit{maximum}) independent set of $G$.

\hskip 10pt $G$ is called a \textit{K\"{o}nig-Egerv\'{a}ry graph} if its order equals
$\alpha(G)+\mu(G)$, where $\mu(G)$ denotes the size of a maximum matching. By
\textrm{core}$(G)$ we mean the intersection of all maximum independent sets
of $G$.

\hskip 10pt To decide whether \textrm{core}$(G)=\emptyset$ is known to be
\textbf{NP}-hard \cite{BGL2002}.

\hskip 10pt In this paper, we present some polynomial time algorithms finding
\textrm{core}$(G)$ of a K\"{o}nig-Egerv\'{a}ry graph $G$.
\end{abstract}

\section{Introduction}

\hskip 10pt Throughout this paper $G=(V,E)$ is a finite, undirected, loopless and without
multiple edges graph with vertex set $V=V(G)$ of cardinality $\left\vert
V\left(  G\right)  \right\vert =n$, and edge set $E=E(G)$ of cardinality
$\left\vert E\left(  G\right)  \right\vert =m$.

If $X\subset V$, then $G[X]$ is the subgraph of $G$ spanned by $X$. By $G-W$
we mean the subgraph $G[V-W]$, if $W\subset V(G)$. The neighborhood of a
vertex $v\in V$ is the set $N(v)=\{w:w\in V$ \ \textit{and} $vw\in E\}$, while
$N(A)=\cup\{N(v):v\in A\}$ and $N[A]=A\cup N(A)$ for $A\subset V$.

A set $S\subseteq V(G)$ is \textit{independent} if no two vertices from $S$
are adjacent; by $\mathrm{Ind}(G)$ we mean the set of all the independent sets
of $G$. An independent set of maximum size will be referred to as a
\textit{maximum independent set} of $G$, and the \textit{independence number
}of $G$ is $\alpha(G)=\max\{\left\vert S\right\vert :S\in\mathrm{Ind}(G)\}$.
In the sequel, the family $\{S:S$ \textit{is a maximum independent set of} $G\}$
is denoted by $\Omega(G)$.

A \textit{matching} in a graph $G=(V,E)$ is a set $M\subseteq E$ such
that no two edges of $M$ share a common vertex. A matching of maximum
cardinality $\mu(G)$ is a \textit{maximum matching}, and a \textit{perfect
matching} is one covering all vertices of $G$.

It is known that $\alpha(G)+\mu(G)\leq\left\vert V\left(  G\right)
\right\vert $. If $\alpha(G)+\mu(G)=\left\vert V\left(  G\right)  \right\vert
$, then $G$ is called a \textit{K\"{o}nig-Egerv\'{a}ry graph }(Deming
\cite{dem}, and Sterboul \cite{ster}). It is easy to see that if $G$ is a
K\"{o}nig-Egerv\'{a}ry\emph{ }graph, then $\alpha(G)\geq\mu(G)$, and that a
graph $G$ having a perfect matching is a K\"{o}nig-Egerv\'{a}ry\emph{ }graph
if and only if $\alpha(G)=\mu(G)$.

K\"{o}nig-Egerv\'{a}ry graphs were investigated in several papers, among we quote
\cite{bourpull,korach,levm2,levm4,LevMan3,lov,lovpl,pulleybl}, and generalized
in \cite{BouHamSim1984,pasdema}.

According to a celebrated result of K\"{o}nig \cite{koen}, and Egerv\'{a}ry
\cite{eger}, every bipartite graph is a K\"{o}nig-Egerv\'{a}ry\emph{ }graph.
This class includes non-bipartite graphs as well (see, for instance, the
graphs $H_{1}$ and $H_{2}$ in Figure \ref{fig1}).

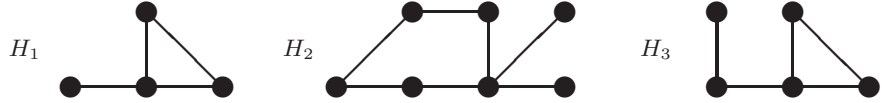
\begin{figure}[h]
\setlength{\unitlength}{1cm}
\begin{picture}(5,1.3)\thicklines
\multiput(1,0)(1,0){3}{\circle*{0.29}}
\put(2,1){\circle*{0.29}}
\put(1,0){\line(1,0){2}}
\put(2,0){\line(0,1){1}}
\put(3,0){\line(-1,1){1}}
\put(0.4,0.5){\makebox(0,0){$H_{1}$}}
\multiput(4.5,0)(1,0){4}{\circle*{0.29}}
\put(6.5,1){\circle*{0.29}}
\put(5.5,1){\circle*{0.29}}
\put(7.5,1){\circle*{0.29}}
\put(4.5,0){\line(1,0){3}}
\put(5.5,1){\line(1,0){1}}
\put(6.5,0){\line(0,1){1}}
\put(4.5,0){\line(1,1){1}}
\put(6.5,0){\line(1,1){1}}
\put(4,0.5){\makebox(0,0){$H_{2}$}}
\multiput(9.5,0)(1,0){3}{\circle*{0.29}}
\multiput(9.5,1)(1,0){2}{\circle*{0.29}}
\put(9.5,0){\line(1,0){2}}
\put(9.5,0){\line(0,1){1}}
\put(10.5,0){\line(0,1){1}}
\put(10.5,1){\line(1,-1){1}}
\put(8.7,0.5){\makebox(0,0){$H_{3}$}}
\end{picture}\caption{Only $H_{3}$ is not a K\"{o}nig--Egerv\'{a}ry graph, as
$\alpha(H_{3})+\mu(H_{3})=4<5=\left\vert V(H_{3})\right\vert $.}
\label{fig1}
\end{figure}

A characterization of K\"{o}nig-Egerv\'{a}ry graphs has been found
independently by Deming \cite{dem} and Sterboul \cite{ster}. Recently, it has
been presented a forbidden subgraph characterization of K\"{o}nig-Egerv\'{a}ry
graphs \cite{KoNgPeis}. Other characterizations of K\"{o}nig-Egerv\'{a}ry
graphs can be found in \cite{Larson2009,LevManKE091,LevManKE092}.

\begin{theorem}
\label{th1}\cite{Gavril}, \cite{dem} Given a graph $G$ and a maximum matching
of $G$, one can test whether $G$ is a K\"{o}nig-Egerv\'{a}ry graph in time $O(m+n)$.
\end{theorem}

\begin{theorem}
\label{th2}\cite{Vazi} Given a graph $G$, one can find a maximum matching in
time $O(m\bullet\sqrt{n})$.
\end{theorem}

As a consequence of Theorems \ref{th1}, \ref{th2} one can deduce the following.

\begin{corollary}
\label{cor1}Given a graph $G$, one can
check in time $O(m\bullet\sqrt{n})$ whether $G$ is a K\"{o}nig-Egerv\'{a}ry graph.
\end{corollary}

Let us recall that \textrm{core}$(G)=\cap\{S:S\in\Omega(G)\}$, \cite{levm3}.

\begin{proposition}
\label{prop1}\cite{levm3} For a connected bipartite graph $G=(A,B,E)$ of order
at least two, the following assertions are true:

\emph{(i)} $\alpha(G)>|V(G)|/2$ if and only if $\left\vert \mathrm{core}
(G)\right\vert \geq2$;

\emph{(ii)} $\alpha(G)=|V(G)|/2$ if and only if $\left\vert \mathrm{core}
(G)\right\vert =0$ and $A,B\in\Omega(G)$.
\end{proposition}

Notice that Proposition \ref{prop1}\emph{(i)} is not true for non-bipartite
K\"{o}nig-Egerv\'{a}ry graphs; e.g., the graph $G_{2}$ from Figure \ref{fig3}.

\begin{figure}[h]
\setlength{\unitlength}{1cm}
\begin{picture}(5,1.3)\thicklines
\multiput(3,0)(1,0){4}{\circle*{0.29}}
\multiput(3,1)(1,0){3}{\circle*{0.29}}
\multiput(3,0)(1,0){3}{\line(0,1){1}}
\put(3,1){\line(1,0){1}}
\put(3,0){\line(1,0){3}}
\put(4.7,1){\makebox(0,0){$u$}}
\put(6,0.3){\makebox(0,0){$v$}}
\put(2,0.5){\makebox(0,0){$G_{1}$}}
\multiput(8,0)(1,0){3}{\circle*{0.29}}
\multiput(8,1)(1,0){3}{\circle*{0.29}}
\multiput(8,0)(1,0){3}{\line(0,1){1}}
\put(8,0){\line(1,0){2}}
\put(8,0){\line(0,1){1}}
\put(8,0){\line(1,1){1}}
\put(8,1){\line(1,-1){1}}
\put(9,0){\line(1,1){1}}
\put(9,1){\line(1,-1){1}}
\put(7.7,1){\makebox(0,0){$a$}}
\put(9.3,1){\makebox(0,0){$b$}}
\put(10.3,1){\makebox(0,0){$c$}}
\put(7,0.5){\makebox(0,0){$G_{2}$}}
\end{picture}\caption{$G_{1}$ has $\alpha(G_{1})=4>\left\vert V(G_{1}
)\right\vert /2$ and $\mathrm{core}(G_{1})=\{u,v\}$, while the graph $G_{2}$ has
$\alpha(G_{2})=3=\left\vert V(G_{2})\right\vert /2$ and $\left\vert
\mathrm{core}(G_{1})\right\vert =\left\vert \{a,b,c\}\right\vert \geq2$.}
\label{fig3}
\end{figure}
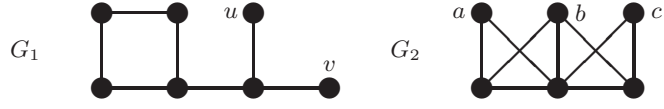

\begin{theorem}
\cite{levm3} For a connected K\"{o}nig-Egerv\'{a}ry graph $G=(V,E)$ of order
at least two, the following assertions are true:

\emph{(i)} $\alpha(G)>|V(G)|/2$ if and only if $\left\vert \mathrm{core}%
(G)\right\vert >\left\vert N(core(G))\right\vert \geq1$;

\emph{(ii)} $\alpha(G)=|V(G)|/2$ if and only if $G$ has a perfect matching.
\end{theorem}

It is known, \cite{BGL2002}, that if $G$ has no isolated vertices, then

\[
\alpha(G)>\mu(G) \Rightarrow \left\vert \mathrm{core}%
(G)\right\vert >\alpha(G)-\mu(G).
\]

Moreover, if
$G$ is a connected graph satisfying $3\mu(G)<\left\vert V(G)\right\vert $, then $G$ has
$\left\vert \mathrm{core}(G)\right\vert \geq2$, \cite{BGL2002}.

\begin{theorem}
\cite{BGL2002} The problem of whether there are vertices in a given graph $G$ belonging to
\textrm{core}$(G)$ is \textbf{NP}-hard.
\end{theorem}
It has been noticed in
\cite{BGL2002} that if $\digamma$ is a hereditary (i.e., induced subgraph
closed) family of graphs for which computing the independence number
$\alpha\left(  G\right)  $ is polynomial, then \textrm{core}$(G)$ can be
computed efficiently for $G\in\digamma$. For instance, it is true in the case
of perfect graphs, line graphs, circular graphs, and circular arc graphs. A
sketch of a sequential algorithm computing \textrm{core}$(G)$ for
K\"{o}nig-Egerv\'{a}ry graphs has been presented in \cite{ChCh2008}.

In this paper, we provide both sequential and parallel algorithms finding
\textrm{core}$(G)$ in polynomial time, where $G$ is a K\"{o}nig-Egerv\'{a}ry graph.

\section{Results and Algorithms}

The following result plays a key role in building our algorithms.

\begin{theorem}
\label{TH}Let $G=(V,E)$ be a K\"{o}nig-Egerv\'{a}ry graph of order $n$, and $v\in V$.

\begin{enumerate}
\item If $\mu(G)=\mu(G-v)$, then $G-v$ is a K\"{o}nig-Egerv\'{a}ry graph
and $v\in\mathrm{core}(G)$.

\item If $\mu(G)=\mu(G-v)+1$, then $G-v$ is a K\"{o}nig-Egerv\'{a}ry graph if
and only if $v\notin\mathrm{core}(G)$.
\end{enumerate}
\end{theorem}

\begin{proof}
By definition of $\mathrm{core}$ : $v\in\mathrm{core}(G)$ if and only if
$\alpha(G)=\alpha(G-v)+1$.

Clearly $\alpha(G-v)+\mu(G-v)\leq n-1$, $\alpha(G)-1\leq \alpha(G-v)\leq \alpha(G)$ and
$\mu(G)-1\leq \mu(G-v)\leq \mu(G)$ hold for every $v\in V(G)$.

\emph{Case 1.} $\mu(G)=\mu(G-v)$.
Assume, to the contrary, that $G-v$ is not a K\"{o}nig-Egerv\'{a}ry graph. Hence $G-v$
satisfies the inequality
\[
\alpha(G-v)+\mu(G-v)<n-1=\alpha(G)+\mu(G)-1,
\]
which leads to the following contradiction: $\alpha(G-v)<\alpha(G)-1$.
Therefore, $G-v$ is a K\"{o}nig-Egerv\'{a}ry graph, and, moreover, we infer that $\alpha(G-v)=\alpha(G)-1$,
 i.e., $v\in\mathrm{core}(G)$.

\emph{Case 2.} $\mu(G)=\mu(G-v)+1$. Then $G-v$ is a K\"{o}nig-Egerv\'{a}ry
graph if and only if
\[
\alpha(G-v)+\mu(G-v)=n-1=\alpha(G)+\mu(G)-1\Longleftrightarrow\alpha
(G)=\alpha(G-v),
\]
i.e., $v\notin\mathrm{core}(G)$, and this completes the proof.

\end{proof}

Taking into account that every subgraph of a bipartite graph is bipartite, one can see that Theorem \ref{TH}
is specified as follows.

\begin{corollary}
Let $G=(V,E)$ be a bipartite graph and $v\in V$. Then $v\in\mathrm{core}(G)$
if and only if $\mu(G)=\mu(G-v)$.
\end{corollary}

Let us notice that if $G$ is a K\"{o}nig-Egerv\'{a}ry graph and has a perfect matching,
then $\mu(G)=\mu(G-v)+1$ holds for every $v\in V(G)$.
Hence by Theorem \ref{TH} we deduce the following.

\begin{corollary}
Let $G=(V,E)$ be a K\"{o}nig-Egerv\'{a}ry graph with a perfect matching
and $v\in V$. Then $v\in\mathrm{core}(G)$ if and only if $G-v$ is not a
K\"{o}nig-Egerv\'{a}ry graph.
\end{corollary}

Theorem \ref{TH} motivates the subsequent algorithm finding $\mathrm{core}
(G)$ for a general K\"{o}nig-Egerv\'{a}ry graph $G$.

\begin{algorithm}
\label{A}

Input = a K\"{o}nig-Egerv\'{a}ry graph $G=(V,E)$

\vskip5pt
Output = $\mathrm{core}(G)=\bigcup\limits_{c(v)=1}\{v\}$, where $c\left(  v\right)  =\left\{
\begin{array}
[c]{c}
1\ \emph{, if} \ v\in\mathrm{core}(G)\\
0\ \emph{, if} \ v\notin\mathrm{core}(G)
\end{array}
\right.$

\vskip5pt
Sequential Complexity =
\[
O(m\bullet\sqrt{n})+O(n\ast\left(  m\bullet\sqrt{n}+\left(  m+n\right)
\right)  )=O(m\bullet n\bullet\sqrt{n})
\]

\vskip5pt

Parallel Complexity with $n$ processors =
\[
O(m\bullet\sqrt{n})+O(m\bullet\sqrt{n}+\left(  m+n\right)  )=O(m\bullet
\sqrt{n})
\]

\vskip4pt

\begin{enumerate}
\item compute $\mu(G)$

\vskip 3pt
\item \textbf{for all} $v\in V$ \textbf{do in parallel}

\vskip 3pt
\item \qquad compute $\mu(G-v)$

\vskip 3pt
\item \qquad \textbf{if} $\mu(G)=\mu(G-v)$

\vskip 3pt
\item \qquad\qquad \textbf{then} $c(v):=1$

\vskip 3pt
\item \qquad\qquad \textbf{else} compute $ke(v):=$ $G-v$ \textbf{is} a
K\"{o}nig-Egerv\'{a}ry graph

\vskip 3pt
\item \qquad\qquad\qquad $c(v):=\overline{ke(v)}$

\vskip 3pt
\item $\mathrm{core}(G):=\bigcup\limits_{c(v)=1}\{v\}$
\end{enumerate}
\end{algorithm}

\begin{figure}[h]
\setlength{\unitlength}{1cm}
\begin{picture}(5,1.3)\thicklines
\multiput(2,0)(1,0){5}{\circle*{0.29}}
\multiput(3,1)(2,0){2}{\circle*{0.29}}
\multiput(3,0)(2,0){2}{\line(0,1){1}}
\put(2,0){\line(1,1){1}}
\put(2,0){\line(1,0){4}}
\put(1.6,0){\makebox(0,0){$v_{1}$}}
\put(2.6,1){\makebox(0,0){$v_{2}$}}
\put(3.3,0.3){\makebox(0,0){$v_{3}$}}
\put(5.3,0.3){\makebox(0,0){$v_{4}$}}
\put(4.3,0.3){\makebox(0,0){$v_{5}$}}
\put(6.1,0.3){\makebox(0,0){$v_{6}$}}
\put(5.35,1){\makebox(0,0){$v_{7}$}}
\put(1,0.5){\makebox(0,0){$G_{1}$}}
\multiput(8.2,0)(1,0){4}{\circle*{0.29}}
\multiput(8.2,1)(1,0){3}{\circle*{0.29}}
\multiput(8.2,0)(1,0){3}{\line(0,1){1}}
\put(8.2,1){\line(1,0){1}}
\put(8.2,0){\line(1,0){3}}
\put(7.85,0){\makebox(0,0){$x_{1}$}}
\put(7.85,1){\makebox(0,0){$x_{2}$}}
\put(9.5,0.3){\makebox(0,0){$x_{3}$}}
\put(9.55,1){\makebox(0,0){$x_{4}$}}
\put(10.5,0.3){\makebox(0,0){$x_{5}$}}
\put(10.55,1){\makebox(0,0){$x_{6}$}}
\put(11.3,0.3){\makebox(0,0){$x_{7}$}}
\put(7.1,0.5){\makebox(0,0){$G_{2}$}}
\end{picture}\caption{$G_{1},G_{2}$ are K\"{o}nig-Egerv\'{a}ry graphs without
a perfect matching.}
\label{fig4}
\end{figure}
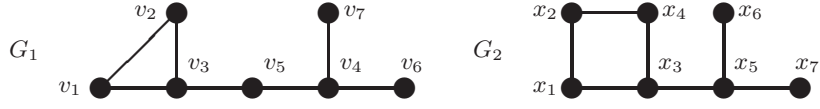

For instance, applying Algorithm \ref{A} for the graph $G_{1}$ from Figure \ref{fig4},
we get the following:

\begin{itemize}
\item $\mu(G_{1})=3$

\item $\mu(G_{1}-v_{i})=3=\mu(G_{1})$ for $i\in\{6,7\}$

\item $\mu(G_{1}-v_{i})=2<\mu(G_{1})$ for $i\in\{1,...,5\}$

\item $c(v_{i}):=1$ for $i\in\{6,7\}$

\item $ke(v_{i})=1$, i.e., $G_{1}-v_{i}$ \textbf{is} a
K\"{o}nig-Egerv\'{a}ry graph for $i\in\{1,2,3,4\}$

\item $ke(v_{i})=0$, i.e., $G_{1}-v_{i}$ \textbf{is not} a
K\"{o}nig-Egerv\'{a}ry graph, for $i\in\{5\}$

\item $c(v_{i})=0$ for $i\in\{1,...,4\}$

\item $c(v_{5}):=\overline{ke(v_{5})}=1$

\item consequently, $\mathrm{core}(G_{1})=\{v_{5},v_{6},v_{7}\}$.
\end{itemize}

\begin{proposition}
Algorithm \ref{A} correctly computes $\mathrm{core}(G)$ of a
K\"{o}nig-Egerv\'{a}ry graph $G$ on $n$ vertices and $m$ edges, with

\emph{(i)} sequential time complexity $O(m\bullet n\bullet\sqrt{n})$;

\emph{(ii)} parallel time complexity with $n$ processors $O(m\bullet\sqrt{n})$.
\end{proposition}

\begin{proof}
According to Theorem \ref{TH}, to decide whether a vertex $v\in V(G)$ belongs
or not to $\mathrm{core}(G)$, one has:

\begin{enumerate}
\item to compute $\mu(G)$, and this requires $O(m\bullet\sqrt{n})$ time,
by Theorem \ref{th2};

\item to compute $\mu(G-v)$ and a maximum matching $M$ of $G-v$, which can
be performed in $O(m\bullet\sqrt{n})$ time, according to Theorem \ref{th2};

\item to check whether $G-v$ that has $M$ as a maximum matching, is a
K\"{o}nig-Egerv\'{a}ry graph or not, and this test can be done in
$O(m+n)$ time, in accordance with Theorem \ref{th1}.
\end{enumerate}

Consequently, the sequential time complexity of Algorithm \ref{A} is
\[
O(m\bullet\sqrt{n})+O(n\ast\left(  m\bullet\sqrt{n}+\left(  m+n\right)
\right)  )=O(m\bullet n\bullet\sqrt{n}),
\]
while its parallel time complexity with $n$ processors is
\[
O(m\bullet\sqrt{n})+O(m\bullet\sqrt{n}+\left(  m+n\right)  )=O(m\bullet
\sqrt{n}),
\]
as claimed.
\end{proof}

If the input graph $G$ is bipartite, then $ke(v)=1$, for every $v\in V(G)$,
because $G-v$ is always bipartite, hence a K\"{o}nig-Egerv\'{a}ry graph.
Consequently, for bipartite graphs we obtain the following simpler algorithm.

\begin{algorithm}
\label{B}

Input = a bipartite graph $G=(V,E)$

Output = $\mathrm{core}(G)=\bigcup\limits_{c(v)=1}\{v\}$, where $c\left(  v\right)  =\left\{
\begin{array}
[c]{c}
1\ \emph{, if} \ v\in\mathrm{core}(G)\\
0\ \emph{, if} \ v\notin\mathrm{core}(G)
\end{array}
\right.$

\vskip 5pt

Sequential Complexity =
\[
O(m\bullet\sqrt{n})+O(n\ast\left(  m\bullet\sqrt{n}\right)  )=O(m\bullet
n\bullet\sqrt{n})
\]

Parallel Complexity with $n$ processors =
\[
O(m\bullet\sqrt{n})+O(m\bullet\sqrt{n})=O(m\bullet\sqrt{n})
\]

\begin{enumerate}
\item compute $\mu(G)$

\item \textbf{for all} $v\in V$ \textbf{do in parallel}

\item \qquad compute $\mu(G-v)$

\item \qquad \textbf{if} $\mu(G)=\mu(G-v)$

\item \qquad\qquad \textbf{then} $c(v):=1$

\item \qquad\qquad \textbf{else} $c(v):=0$

\item $\mathrm{core}(G):=\bigcup\limits_{c(v)=1}\{v\}$
\end{enumerate}
\end{algorithm}

For example, applying Algorithm \ref{B} to the graph $G_{2}$ depicted in Figure \ref{fig4}, we
obtain the following:

\begin{itemize}
\item $\mu(G_{2})=3$

\item $\mu(G_{2}-x_{i})=2<\mu(G_{2})$ for $i\in\{1,...,5\}$

\item $\mu(G_{2}-x_{i})=3=\mu(G_{2})$ for $i\in\{6,7\}$

\item $c(x_{i})=0$ for $i\in\{1,...,5\}$

\item $c(x_{i})=1$ for $i\in\{6,7\}$

\item consequently, $\mathrm{core}(G_{2})=\{x_{6},x_{7}\}$.
\end{itemize}

Let us notice that, unlike bipartite graphs, a K\"{o}nig-Egerv\'{a}ry graph
$G$ with a perfect matching can have $\mathrm{core}(G)\neq\emptyset$; e.g.,
the graphs $H_{1}$ and $H_{2}$ from Figure \ref{fig111} have at least one
perfect matching and $\mathrm{core}(H_{1})=\{x\}$, while $\mathrm{core}%
(H_{2})=\{u,v\}$.

\begin{figure}[h]
\setlength{\unitlength}{1cm}
\begin{picture}(5,1.3)\thicklines
\multiput(3,0)(1,0){4}{\circle*{0.29}}
\multiput(4,1)(1,0){2}{\circle*{0.29}}
\put(3,0){\line(1,0){3}}
\put(4,0){\line(0,1){1}}
\put(4,1){\line(1,0){1}}
\put(5,1){\line(1,-1){1}}
\put(2.7,0){\makebox(0,0){$x$}}
\put(2,0.5){\makebox(0,0){$H_{1}$}}
\multiput(8,0)(1,0){3}{\circle*{0.29}}
\multiput(8,1)(1,0){3}{\circle*{0.29}}
\put(8,0){\line(1,0){2}}
\put(8,0){\line(0,1){1}}
\put(8,1){\line(1,0){1}}
\put(9,0){\line(0,1){1}}
\put(9,0){\line(1,1){1}}
\put(10,0){\line(0,1){1}}
\put(8.25,0.3){\makebox(0,0){$u$}}
\put(9.3,1){\makebox(0,0){$v$}}
\put(7,0.5){\makebox(0,0){$H_{2}$}}
\end{picture}\caption{$H_{1}$ and $H_{2}$ are K\"{o}nig-Egerv\'{a}ry graphs
with perfect matchings.}
\label{fig111}
\end{figure}
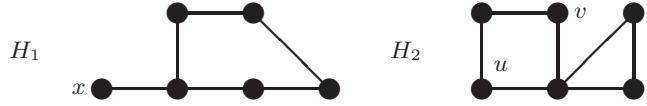

If $G$ is a K\"{o}nig-Egerv\'{a}ry graph having a perfect matching, then clearly,
$\mu(G)=\mu(G-v)+1$ holds for every $v\in V(G)$. Hence, $v\in\mathrm{core}(G)$
if and only if $G-v$ is not a K\"{o}nig-Egerv\'{a}ry graph.
Consequently, $\mathrm{core}(G)$ of a K\"{o}nig-Egerv\'{a}ry graph $G$
owning a perfect matching, may be found more efficiently.

\begin{algorithm}
\label{C}

Input = a K\"{o}nig-Egerv\'{a}ry graph $G$ with a perfect matching

Output = $\mathrm{core}(G)=\bigcup\limits_{c(v)=1}\{v\}$, where $c\left(  v\right)  =\left\{
\begin{array}
[c]{c}
1\ \emph{, if} \ v\in\mathrm{core}(G)\\
0\ \emph{, if} \ v\notin\mathrm{core}(G)
\end{array}
\right.$

\vskip 3pt

Sequential Complexity = $O(n\bullet\left(  m\bullet\sqrt{n}\right)
)=O(m\bullet n\bullet\sqrt{n})$

\vskip 3pt

Parallel Complexity with $n$ processors = $O(m\bullet\sqrt{n})$

\begin{enumerate}
\item \textbf{for all} $v\in V(G)$ \textbf{do in parallel}

\item \qquad compute $ke(v):=$ $G-v$ \textbf{is} a
K\"{o}nig-Egerv\'{a}ry graph

\item \qquad$c(v):=\overline{ke(v)}$

\item $\mathrm{core}(G):=\bigcup\limits_{c(v)=1}\{v\}$
\end{enumerate}
\end{algorithm}

Applying Algorithm \ref{C} for the graph $G_{1}$ from Figure \ref{Fig1}, we
get that: $ke(v_{1})=0$ and $ke(v_{3})=0$, i.e., $G_{1}-v_{1}$ and
$G_{1}-v_{3}$ are not K\"{o}nig-Egerv\'{a}ry graphs, while
$ke(v_{2})=ke(v_{4})=ke(v_{5})=ke(v_{6})=1$, i.e., $G_{1}-v_{i},i\in
\{2,4,5,6\}$, are still K\"{o}nig-Egerv\'{a}ry graphs. Consequently, it
follows that only $c(v_{1})=c(v_{3})=1$, and hence $\mathrm{core}(G_{1})=\{v_{1},v_{3}\}$.

\begin{figure}[h]
\setlength{\unitlength}{1cm}
\begin{picture}(5,1.2)\thicklines
\multiput(3,0)(1,0){4}{\circle*{0.29}}
\multiput(3,1)(3,0){2}{\circle*{0.29}}
\multiput(3,0)(3,0){2}{\line(0,1){1}}
\put(5,0){\line(1,1){1}}
\put(3,0){\line(1,0){3}}
\put(2.6,1){\makebox(0,0){$v_{1}$}}
\put(2.6,0){\makebox(0,0){$v_{2}$}}
\put(4,0.35){\makebox(0,0){$v_{3}$}}
\put(5,0.35){\makebox(0,0){$v_{4}$}}
\put(6.4,1){\makebox(0,0){$v_{5}$}}
\put(6.4,0){\makebox(0,0){$v_{6}$}}
\put(1.8,0.5){\makebox(0,0){$G_{1}$}}
\multiput(8.5,0)(1,0){3}{\circle*{0.29}}
\multiput(8.5,1)(1,0){3}{\circle*{0.29}}
\multiput(8.5,0)(1,0){3}{\line(0,1){1}}
\put(9.5,1){\line(1,0){1}}
\put(8.5,0){\line(1,0){2}}
\put(7.5,0.5){\makebox(0,0){$G_{2}$}}
\end{picture}\caption{$G_{1}$ and $G_{2}$ are K\"{o}nig-Egerv\'{a}ry graphs
with perfect matchings.}
\label{Fig1}
\end{figure}
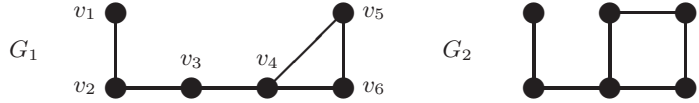

It is worth mentioning that if the input graph $G$ is bipartite having a perfect
matching, then Algorithm \ref{C} gives a constructive proof of Proposition
\ref{prop1}\emph{(ii)} claiming that $\mathrm{core}(G)=\emptyset$. For
example, using Algorithm \ref{C} for the bipartite graph $G_{2}$ from Figure
\ref{Fig1}, one can see that $ke(v)=1$ holds for every $v\in V(G_{2})$, and
hence, $\mathrm{core}(G_{2})=\emptyset$.

\section{Conclusions}

In this paper we present a sequential algorithm with time complexity
$O(m\bullet n^{\frac{3}{2}})$ finding $\mathrm{core}(G)$ of a
K\"{o}nig-Egerv\'{a}ry\emph{ }graph. Its parallel counterpart solves the same
problem in $O(m\bullet n^{\frac{1}{2}})$ time complexity.

It is known that the unique maximum independent set problem is \textbf{NP}-hard for general graphs \cite{pelc}.
One of applications of our results is a polynomial algorithm recognizing a K\"{o}nig-Egerv\'{a}ry graph with
a unique maximum independent set. In fact, the graph $G$ has a unique maximum
independent set if and only if $\mathrm{core}(G)$ is a maximal independent set
\cite{FRV2003}. Therefore, whenever there is a polynomial algorithm returning
$\mathrm{core}(G)$, one can decide in polynomial time whether $G$ has a unique
maximum independent set. Consequently, to recognize a
K\"{o}nig-Egerv\'{a}ry graph $G$ with a unique maximum independent set,
it is enough to run Algorithm \ref{A}, and then to try enlarging its output to
an independent set. The enlarging part is handled in $O(m)$ time complexity
sequentially, while in parallel it may be implemented with $O(1)$ time
complexity.

\end{document}